\documentclass[acmlarge]{acmart}

\AtBeginDocument{%
  }

\setcopyright{acmlicensed}
\copyrightyear{2025}
\acmYear{2025}
\acmDOI{XXXXXXX.XXXXXXX}

\acmJournal{TQC}

\usepackage{xcolor}
\definecolor{notered}{HTML}{ff0000}
\definecolor{RTBlue}{HTML}{3366ff}
\definecolor{MC}{HTML}{8B8000}

\definecolor{red}{HTML}{D62728}
\definecolor{blue}{HTML}{1F77B4}

\usepackage{adjustbox}
\usepackage{braket}
\usepackage{caption}
\usepackage{slashbox}
\usepackage{csvsimple}
\usepackage{anyfontsize}
\usepackage{multirow}
\usepackage{mathtools}
\newcommand{\defeq}{\vcentcolon=}

\usepackage{pgf,tikz}
\DeclareMathOperator*{\argmax}{arg\,max}

\usepackage{amsthm}
\newtheorem{lemma}{Lemma}

\usepackage{hyperref}
\begin{document}

\title[Warm-Start QAOA via a
Reduction to Max-Cut]{Solving General QUBOs with Warm-Start QAOA via a
Reduction to Max-Cut}

\author{Bikrant Bhattacharyya}
\affiliation{%
  \institution{California Institute of Technology}
  \country{USA}}
\email{bbhattac@caltech.edu}

\author{Michael Capriotti}
\affiliation{%
  \institution{Northwestern University}
  \country{USA}}
\email{michaelcapriotti2028@u.northwestern.edu}

\author{Reuben Tate}
\affiliation{%
  \institution{CC3: Information Sciences, Los Alamos National Laboratory}
  \country{USA}}
\email{rtate@lanl.gov}

\renewcommand{\shortauthors}{Bhattacharyya et al.}

\begin{abstract}
  The Quantum Approximate Optimization Algorithm (QAOA) is a quantum algorithm that finds approximate solutions to problems in combinatorial optimization, especially those that can be formulated as a Quadratic Unconstrained Binary Optimization (QUBO) problem. In prior work, researchers have considered various ways of "warm-starting" QAOA by constructing an initial quantum state using classically-obtained solutions or information; these warm-starts typically cause QAOA to yield better approximation ratios at much lower circuit depths. For the Max-Cut problem, one warm-start approaches constructs the initial state using the high-dimensional vectors that are output from an SDP relaxation of the corresponding Max-Cut problem. This work leverages these semidefinite warmstarts for a broader class of problem instances by using a standard reduction that transforms any QUBO instance into a Max-Cut instance. We empirically compare this approach to a "QUBO-relaxation" approach that relaxes the QUBO directly. Our results consider a variety of QUBO instances ranging from randomly generated QUBOs to QUBOs corresponding to specific problems such as the traveling salesman problem, maximum independent set, and portfolio optimization. We find that the best choice of warmstart approach is strongly dependent on the problem type.
\end{abstract}

\ccsdesc[500]{Theory of computation~Semidefinite programming}
\ccsdesc[500]{Hardware~Quantum computation}
\keywords{QAOA, approximation algorithms, Max-Cut, warm-starts}


\maketitle

\section{Introduction}
The Quantum Appoximate Optimization Algorithm (QAOA) is a hybrid quantum-classical optimization algorithm developed by Farhi et al. \cite{farhi2014quantumapproximateoptimizationalgorithm} designed to solve combinatorial optimization problems. This algorithm uses a classical optimization loop to fine-tune the parameters of a quantum circuit. The ansatz for QAOA circuits is based off of a Trotterization of the adiabatic quantum algorithm \cite{farhi2014quantumapproximateoptimizationalgorithm, Wurtz_2022}. \\\\
Many NP-hard problems can be formulated as quadratic optimization problems over discrete variables \cite{Lucas_2014}; such quadratic formulations are well-suited for QAOA (and quantum algorithms in general) due to the natural correspondence between the quadratic terms and native two-qubit gates that exist on most quantum devices.
\\\\
QAOA falls under the broader class of quantum algorithms called Variational Quantum Algorithms \cite{Wang2021}, which all use this hybrid quantum-classical optimization loop to minimize a cost function evaluated on a quantum circuit. 
\\\\
For current and near-term quantum devices, it is important to be mindful of the circuit depth since an increase in circuit depth and gate count as the former will causing to an increase in total runtime leading to increased decoherence and the latter will cause an increase in the compounded gate errors. Tate et al. \cite{Tate_2023,tate2022bridgingclassicalquantumsdp} developed a ``warm-start" heuristic that modifies the initial quantum state of QAOA from $\ket{+}^{\otimes n}$ (the uniform superposition of all bitstrings) to warm-start initial state that is designed to be biased towards better solutions. The idea is that with an improved initial state, fewer QAOA layers (and hence fewer gates and lower circuit depths) are needed to transform the initial state into a quantum state whose measurement yields good solutions, and indeed, this was what has been empirically observed. 
\\\\
In particular, when solving Max-Cut problems, one can relax the problem into a semidefinite program (SDP) as is done in the Goemans-Williamson algorithm \cite{10.1145/227683.227684}; the SDP is a convex problem that can be solved efficiently and Tate et al. uses the solutions from this SDP relaxation (which are high-dimensional vectors) to construct the warm-start initial quantum state.\cite{tate2022bridgingclassicalquantumsdp}.
\\\\
Around the same time that Tate et al. \cite{tate2022bridgingclassicalquantumsdp} developed their warm-start approach, Egger et al. \cite{Egger_2021} had independently developed their own warm-start approaches. For any given QUBO, one of their approaches considers a simple QUBO-relaxation that relaxes the integer constraints $x_i \in \{0,1\}$ of the QUBO to an interval constarint $x_i \in [0,1]$; the optimal solutions from the QUBO-relxation are then used to generate a warm-started quantum state. If the matrix defining the QUBO satisfies certain criteria, the relaxation is convex and can thus be solved efficiently; however, for general QUBOs, solving this relaxation is known to be NP-Hard \cite{pardalos1991quadratic}. For an arbitrary QUBO, one possible approach is to perform the above QUBO-relaxation and to find an \emph{approximate} optimal solution (of the relaxation) via local optimization.
\\\\
Alternatively, this work proposes a new approach: a mapping that takes arbitrary QUBOs and maps them to equivalent Max-Cut instances. Using the warm-start approach of Tate et al. \cite{tate2022bridgingclassicalquantumsdp,Tate_2023}, these Max-Cut problems can be solved to obtain solutions that can then be mapped back to solutions for the original QUBO.
\\\\
This ``QUBO to Max-Cut'' mapping introduces an additional auxiliary qubit/variable. In Tate et al.'s approach, a global rotation is typically applied to the solution of the SDP relaxation before mapping it to a warm-started initial quantum state. They propose a vertex-at-top heuristic where the global rotation rotates one of the qubits so that it is on top of the Bloch sphere in the warm-start quantum state. However, some choices of vertices are better than others, and hence, an $O(n)$ overhead (where $n$ is the number of vertices) is required to test all possible vertex-at-top rotations. Interestingly, when mapping a QUBO to a Max-Cut instance, we find that the auxiliary qubit is frequently the best choice of qubit for the vertex-at-top heuristic; this observation can be used to reduce the overhead of trying multiple random vertices in the vertex-at-top heuristic.
\\\\
We empirically compare the two warm-start approaches above for a variety of different types of QUBOs up to 16 variables. The QUBOs we test come from randomly generated matrices or QUBOs that arise from specific problems in combinatorial optimization (e.g. portfolio optimization, traveling salesman, and maximum independent set). We find that the best warm-start approach is highly problem-dependent and also dependent on the metric of success (e.g. approximation ratio vs probability of finding the optimal solution).

\section{Background}
\subsection{Quadratic Combinatorial Optimization Problems}

Here we only consider two types of (quadratic) problems.
Given a symmetric matrix $Q\in \mathbb{R}^{n\times n}$ the associated \textit{Quadratic Unconstrained Binary Optimzation} (QUBO) problem is \cite{lewis2017quadraticunconstrainedbinaryoptimization,glover2019tutorialformulatingusingqubo}
\begin{equation}
    \max_{x\in\{0,1\}^n}x^T Q x
\end{equation}
Notice that this structure allows for linear terms as well because for $x\in\{0,1\}^n$, $x_i^2 = x_i$ for all $i=1,2,\dots,n$, and hence,
\begin{equation}
x^T Q x + \mu ^ T x = x^T(Q+\textrm{diag}(\mu))x.
\end{equation}
Likewise for a symmetric matrix $J\in \mathbb{R}^{n\times n}$ the corresponding \textit{Ising Problem} is
\begin{equation}\label{Ising Cost}
    \max_{y\in\{-1,1\}^n}y^T J y.
\end{equation}
With this definition, Ising problems, unlike QUBOs, do not allow for linear terms.
\\\\
Ising problems are closely related to the Ising model studied in physics, which is a statistical model for spin couplings in ferromagnetic materials \cite{Cipra1987AnIT,Shankar_2017}.\footnote{In the context of physical systems, the addition of a linear term would correspond to the introduction of an external magnetic field.}\\
\noindent
Here we define a well-studied combinatorial optimization problem: (weighted) Max-Cut \cite{10.1145/502090.502098,doi:10.1137/S0097539797328847,1366224} where the objective is to partition the vertices of a weighted graph into two disjoint groups so that the sum of the weights of the edges between the groups is maximized. More formally, given a graph $G=(V,E)$ with weighted adjacency matrix $A$, Max-Cut can be written as the following maximization problem:
\begin{equation}
    \max_{y \in \{-1,+1\}^n} C(y),
\end{equation}
with,
\begin{equation}
    C(y)=\frac{1}{4}\sum_{1\leq i,j \leq n}A(1-y_iy_j),
\end{equation}
where the elements of $\{-1,1\}^{|V|}$ are referred to as \textit{cuts} with corresponding \textit{cut-value} of $C(y)$.


This differs by a constant value from the Ising problem associated with $-A/4$,
\begin{equation}
    C(y)=y^T\left(-\frac{1}{4}A\right)y+\frac{1}{4}\textrm{sum}(A).
\end{equation}
We call the Ising problem associated with $-A/4$ the \textit{Max-Cut} problem for graph $G$.
\\\\
The (instance-specific) cut ratio of $y$ is given by 
$$\alpha(y)=\frac{C(y)-\min_{y\in\{-1,1\}^n} C(y)}{\max_{y\in\{-1,1\}^n} C(y)-\min_{y\in\{-1,1\}^n} C(y)}.$$

\subsubsection{From QUBOs to Max-Cut Problems}
Given the QUBO associated with matrix $Q\in \mathbb{R}^{n\times n}$, the $n+1$ adjacency matrix $A$ given by 
\begin{equation}
\begin{aligned}
A_{i,j}&=-Q_{i,j}\\
A_{n,i}&=\sum_{1\leq j \leq n}Q_{i,j}\\
A_{n,n}&=0,
\end{aligned}
\end{equation}
has the property that if $x_i=\frac{1}{2}(1-y_iy_n)$,
\begin{equation}\label{QUBO_to_MAXCUT}
y^T\left(-\frac{1}{4}A\right)y+\frac{1}{4}\textrm{sum}(Q) = x^TQx.
\end{equation}
A derivation of this mapping can be found in Appendix \ref{Map Deriv}. 

\
The graph with adjacency matrix $A$ is called the \textit{corresponding graph} of $Q$. The corresponding graph allows for QUBOs to be mapped to Max-Cut problems. A similar mapping was mentioned in \cite{Dunning2018WhatWB}.

\subsubsection{Max-Cut Relaxations}
The benefit of mapping QUBOs to Max-Cut problems is that we can make use of well-known relaxations.
\\\\
Relaxing the Max-Cut cost function (\ref{Ising Cost}, with $J=-A/4$) by introducing $Y_i\in \mathbb{R}^k$ with $\|Y_i\|=1$ is given by
\begin{equation}
-\frac{1}{4}\sum_{1\le i,j\leq n}^{ }A_{i,j}y_{i}y_{j}\to-\frac{1}{4}\sum_{1\le i,j\leq n}^{ }A_{i,j}\left(Y^{T}Y\right)_{i,j},
\end{equation}
where $Y$ is a matrix with columns $Y_i$.
The Goemans Williamson GW relaxation \cite{10.1145/227683.227684} is derived from the special case where $k=n$. In this case, $Y^TY$ is a positive semidefinite symmetric matrix with unit diagonal. If we replace the matrix $Y^TY$ with a general matrix $M$, then the relaxed QUBO is 
$$-\frac{1}{4}\max_{M\in \mathbb{S}^n} \operatorname{tr}(A M),$$
where $\mathbb{S}^n$ is the set of positive semidefinite symmetric matricies with unit diagonals.
\\\\
Notably, this optimization problem is QUBO-Relaxed, and can thus be exactly solved \cite{Nesterov1994InteriorpointPA}. Given a solution $M$, we find the Cholesky decomposition $M=Y^TY$. 
If we let $k=2$ or $k=3$, then we obtain the Burer-Monteiro BM relaxations, $\textrm{BM}_2$ and $\textrm{BM}_3$ respectively \cite{Burer2003ANP}. 
\\\\
In the case of the $\textrm{BM}_2$ relaxation, we can express each column vector $Y_i$ in polar coordinates, i.e., $Y$ can be parameterized it in terms of $\theta \in [0,2\pi)^n$, giving
\begin{equation}
    Y_i = \left[\cos(\theta_i), \sin(\theta_i) \right]^T.
\end{equation}
Likewise for $\textrm{BM}_3$ relaxation, each $Y_i$ can be expressed in spherical coordinates, i.e., we can take $Y$ and parameterize it in terms of $\theta\in[0,\pi]^n,\phi = [0,2\pi)^n$, giving
\begin{equation}
    Y_i = \left[\sin(\theta_i)\cos(\phi_i), \sin(\theta_i)\sin(\phi_i),\cos(\theta_i) \right]^T.
\end{equation}
Unlike the GW relaxation, solving for $Y$ is no longer a QUBO-Relaxed problem. Instead we use randomized stochastic perturbations with (step size $\eta$) to solve this problem as in \cite{tate2022bridgingclassicalquantumsdp,Tate_2023}.
\subsection{QAOA}
Given a depth $p$, mixing Hamiltonian $H_B$, and a cost Hamiltonian $H_C$, the Quantum Approximate Optimization Algorithm (QAOA) (see \cite{farhi2014quantumapproximateoptimizationalgorithm}) produces 
a state $|\psi(\beta,\gamma)\rangle$ parametrized by $\beta,\gamma\in \mathbb{R}^p$ as
\begin{equation}\label{QAOA State}
   |\psi\rangle = \prod_{m=1}^pe^{-i\beta_m H_B}e^{-i\gamma_m H_C}|\psi_\textrm{init}\rangle.
\end{equation}
Here, the initial state $|\psi_\textrm{init}\rangle$ is always taken to be a separable state, i.e. $|\psi_\textrm{init}\rangle$ can be expressed as
\begin{equation}
|\psi_\textrm{init}\rangle=\bigotimes_{j=1}^n|\psi_j\rangle,
\end{equation}
where $n$ is the number of qubits and each $|\psi_j\rangle$ is a single qubit state.
\\\\
Following the convention of ``QAOA-warmest'' from \cite{Tate_2023}, we choose $H_B$ to be a seperable Hamiltonian where $|\psi_\textrm{init}\rangle$ is the maximum eigenvalue eigentstate of $H_B$. This can be constructed as
\begin{equation}
    H_B=\bigoplus_{j=0}^{N-1}(x_jX+y_jY+z_jZ),
\end{equation}
where $(x_j,y_j,z_j)$ are the bloch sphere coordinates for the (single qubit) state $|\psi_j\rangle$.
\\\\
The goal of QAOA is to estimate the maximimum energy eigenstate of $H_C$ by solving for
$\beta^*,\gamma*$, where
\begin{equation}
    \beta^*,\gamma^* = \argmax_{\beta,\gamma\in \mathbb{R}^p} \langle \psi(\beta,\gamma) | H_C | \psi(\beta,\gamma)\rangle.
\end{equation}
\noindent
Assuming that the depth $p$ QAOA ansatz is sufficiently expressive, the maximal eigenstate of $H_C$ will approximately be 
$|\psi(\beta^*,\gamma^*)\rangle$.
\\\\
In fact, \cite{Tate_2023} showed that\footnote{The theorems by \cite{Tate_2023} specifically prove this limit in the case that $H_C$ is the Max-Cut cost Hamiltonian; however, by going through the proofs, it is clear that the argument holds for any diagonal cost Hamiltonian (e.g. cost Hamiltonians that arise from classical combinatorial optimization problems.) }, just like standard QAOA, if
$H_B$ is constructed as above (where $E_\textrm{max}$ is the maximum eigenvalue of $H_C$), the following  result holds, 
\begin{equation}\label{Infintie Depth Garuntee}
\lim_{p\to\infty} \langle\psi(\beta^*,\gamma^*)|H|\psi(\beta^*,\gamma^*)\rangle=E_\textrm{max},
\end{equation}
as long as none of the qubits in $\ket{\psi_\text{init}}$ lie at the poles of the Bloch sphere.
\subsubsection{Encoding QUBOs and Max-Cut Problems}
Given $x\in\{0,1\}^n,y\in\{-1,1\}^n$ and symmetric $Q,J\in\mathbb{R}^{n\times n}$, we can construct the $n$-qubit states
\begin{equation}
    \begin{aligned}
    |\psi(x)\rangle &= \bigotimes_{m=0}^{n-1}|x_i\rangle,\\
    |\psi(y)\rangle &= \bigotimes_{m=0}^{n-1}\left|\frac{y_i+1}{2}\right\rangle.
    \end{aligned}
\end{equation}
We can also define Hamiltonians\footnote{Here $Q$ and $M$ refer to QUBO and Max-Cut respectively.} $H_Q^{\textrm{QUBO}}$ and $H_J^{\textrm{Max-Cut}}$ as (where $Z_i$ is the Pauli $Z$ applied to qubit $I$)
\begin{equation}
    \begin{aligned}
        H_Q^{\textrm{Q}}& = \sum_{0\leq i,j < n}Q_{i,j}\left(\frac{1-Z_i}{2}\right)\left(\frac{1-Z_i}{2}\right),\\
        H_A^{\textrm{M}}&= -\frac{1}{4}\sum_{0\leq i,j < n}A_{i,j}Z_iZ_j.\\
    \end{aligned}
\end{equation}
Which in turn satisfy
\begin{equation}
    \begin{aligned}
        \langle \psi(x) | H_Q^{\textrm{Q}}| \psi(x)\rangle &= x^TQx,\\
        \langle \psi(y) | H_A^{\textrm{M}}| \psi(y)\rangle &= -\frac{1}{4}y^TAy.
    \end{aligned}
\end{equation}

Now, given the QUBO problem associated with $Q$, it can be solved by either
\begin{itemize}
    \item Running QAOA with $H_C = H^{\textrm{Q}}_Q.$
    \item Running QAOA with $H_C = H^{\textrm{M}}_A$ where $A$ is the adjacency matrix of the corresponding graph of $Q$.
\end{itemize}

\section{Methods}

\subsection{Warm-start Quantum Optimization}
Warm-start QAOA refers to using information about the Hamiltonian $H_C$ to determine an appropriate initial state, $|\psi_\textrm{init}\rangle$.
\subsubsection{QUBO-Relaxed}
The QUBO-Relaxed warmstart (see \cite{Egger_2021,Tate_2023}) is inspired by the following relaxation of a QUBO. Given a symmetric $Q\in \mathbb{R}^{n\times n}$, let
\begin{equation}
    y^c = \argmax_{y\in [0,1]^n} y^TQy.
\end{equation}
\noindent
In the event that $Q$ is negative semidefinite, it is known that the above relaxation is convex\footnote{Technically, the objective is a \emph{concave} function; however convex optimization is usually considered in the context of minimization problems and minimization of a convex objective $f$ is equivalent to maximization of the concave objective $-f$. In general, we abuse this language and say that an optimization problem ``is convex" whenever a convex objective is being minimized or when a concave objective is being maximized.} and thus can be solved efficiently \cite{gondzio2006solving}. However, for arbitrary $Q$, solving the relaxation is known to be NP-Hard \cite{pardalos1991quadratic} meaning that, unless $P=NP$, there is no algorithm that is guaranteed to find the global optimum in polynomial time. As discussed later in our numerical results, for general $Q$, we instead \emph{estimate} $y^c$ by considering random initial points in the box $[0,1]^n$ and performing a local optimization.

Given a positive real hyperparameter $\varepsilon>0$, the QUBO-Relaxed initial state for $H_Q^{\textrm{QUBO}}$ is given by
\begin{equation}
    \begin{aligned}
        |\psi_j\rangle &= R_Y(\theta_j)|0\rangle \quad \textrm{where}\\
        \theta_j &= \begin{cases}
            2\arcsin(\varepsilon) &\textrm{if}\quad y^c_j \leq \varepsilon \\ 
            2\arcsin(y^c_j) &\textrm{ if} \quad\varepsilon \leq y^c_j \leq 1-\varepsilon \\
            2\arcsin(1-\varepsilon) &\textrm{if}\quad y^c_j \geq 1-\varepsilon \\
        \end{cases}.
    \end{aligned}
\end{equation}
\subsubsection{Bloch Sphere Encoding}
Given a mapping with $k=2$ and a corresponding vector of angles $\theta\in [0,2\pi)$, the corresponding initial state for $H^\textrm{Ising}_J$ is
\begin{equation}
    |\psi_j\rangle =  \cos\left(\frac{\theta_j}{2}\right)|0\rangle +  e^{-i\pi/2}\sin\left(\frac{\theta_j}{2}\right)\ket{1}.
\end{equation}
And, given a mapping with $k=3$ and a corresponding vector of angles $\theta\in [0,\pi],\phi\in[0,2\pi)^n$, the corresponding initial state for $H^\textrm{Ising}_J$ is
\begin{equation}
    |\psi_j\rangle =  \cos\left(\frac{\theta_j}{2}\right)|0\rangle +  e^{i\phi_j}\sin\left(\frac{\theta_j}{2}\right)\ket{1}.
\end{equation}
We also use ``vertex-at-top'' rotations defined as follows (see \cite{tate2022bridgingclassicalquantumsdp,Tate_2023})
\begin{itemize}
    \item If $k=2$, a vertex-at-top rotation on qubit $j$ applies the transformation $\theta\to\theta-\theta_j.$
    \item If $k=3$, a vertex-at-top rotation on qubit $j$ applies a clockwise $z$-rotation by $\phi_j$, followed by a clockwise $y$-rotation by $\theta_j$, followed by a random $z$-rotation.
\end{itemize}
It can be easily verified that applying a vertex-at-top rotation to qubit $j$ modifies the angles such that $|\psi_j\rangle=|0\rangle$.
\subsubsection{Dimension \texorpdfstring{$k$}{TEXT} Goemans Williamson (\texorpdfstring{$\textrm{GW}_k$}{TEXT})}
While it's clear how to construct initial states from the $\textrm{BM}_k$ relaxations, constructing initial states from the GW relaxation requires an additional step.
\\\\
Following \cite{Tate_2023}, for $k=2,3$ we sample random orthonormal bases $\{x_i\}$ of $\mathbb{R}^n$ and keep the first $k$. We then project each vector $Y_i \in \mathbb{R}^n$ onto the subspace generated by the basis $\{x_i\}$, i.e., we define $\tilde Y_i\in \mathbb{R}^k$ with components
\begin{equation}
    (\tilde Y_i)_j=\frac{x_j^TY_i}{\sqrt{\sum_{j=0}^{k-1}(x_j^TY_i)^2}}.
\end{equation}
This then gives a sequence of $n$ $k$-dimensional vectors as desired (or alternatively $\tilde Y\in\mathbb{R}^{k\times n}$).
\\\\
Because the bases $\{x_i\}$ are randomly sampled, there's no guarantee that a given one will sufficiently recover the structure of $\{Y_i\}$, we randomly generate a fixed number of bases and choose the one that maximizes $\textrm{tr}(J^T, \tilde Y^T\tilde Y)$.
\subsection{Performance Metrics}
To do so, we evaluate two metrics. The first is inspired by the cut ratio 
\begin{equation}
    \begin{aligned}
        \alpha^\textrm{Q} &= \frac{\langle \psi^{\textrm{Q}}|H^\textrm{Q}_Q|\psi^{\textrm{Q}}\rangle-E^\textrm{Q}_\textrm{min}}{E^\textrm{Q}_\textrm{max}-E^\textrm{Q}_\textrm{min}},\\
        \alpha^\textrm{M} &= \frac{\langle \psi^{\textrm{M}}|H^\textrm{M}_A|\psi^{\textrm{M}}\rangle-E^\textrm{M}_\textrm{min}}{E^\textrm{M}_\textrm{max}-E^\textrm{M}_\textrm{min}}.
    \end{aligned}
\end{equation}
Notice that this comparison removes the affect of the constant difference in (\ref{QUBO_to_MAXCUT}).
\\\\
The second is the optimal sampling probability,
\begin{equation}
    \hspace{-1cm}
    \begin{aligned}
        \mathcal{P}^\textrm{Q} &= \sum_{|\phi^\textrm{Q}\rangle}|\langle \phi^\textrm{Q} | \psi^\textrm{Q}\rangle|^2 \\ & \textrm{where} \quad \langle \phi^\textrm{Q} | H^\textrm{Q}_A | \phi^\textrm{Q}\rangle = E_\textrm{max}^\textrm{Q},\\
        \mathcal{P}^\textrm{M} &= \sum_{|\phi^\textrm{M}\rangle}|\langle \phi^\textrm{M} | \psi^\textrm{M}\rangle|^2 \\ & \textrm{where} \quad \langle \phi^\textrm{M} | H^\textrm{M}_A | \phi^\textrm{M}\rangle = E^\textrm{M}_\textrm{max}.
    \end{aligned}
\end{equation}
\\\\
Because the superscript on $\alpha,\mathcal{P}$ will be made clear by the warmstart in consideration, it will be dropped from here on.\footnote{Obtaining the optimal solution to either the QUBO or Ising problem by measuring the QAOA output requires sufficiently large $\mathcal{P}$. With that being said, there is no way to directly optimize $\mathcal{P}$ without knowing the optimal solution ahead of time, so QAOA optimizes the expectation directly instead, which is turn optimizes for the approximation ratio $\alpha$.}
\\\\
Because $\alpha,\mathcal{P}\in[0,1]$ by construction, they serve as a comparable performance metric for the different warmstarts applied to either $H^\textrm{Q}_Q$ or $H^\textrm{M}_A$.\footnote{The approximation ratio $\alpha$ is not evenly distributed over states. See Appendix~\ref{alphadistributions}} In principle, one could compute $(\alpha,\mathcal{P})$ for the Max-Cut case first, or first project down onto the QUBO space and then compute $(\alpha,\mathcal{P})$. Because of how these metrics are defined, either procedure for computing the metrics gives the same result, allowing for a fair comparison.
\subsection{Example Problems Instances}\label{setup}
\begin{figure*}
    \hspace*{-2cm}\includegraphics[width=1.2\linewidth]{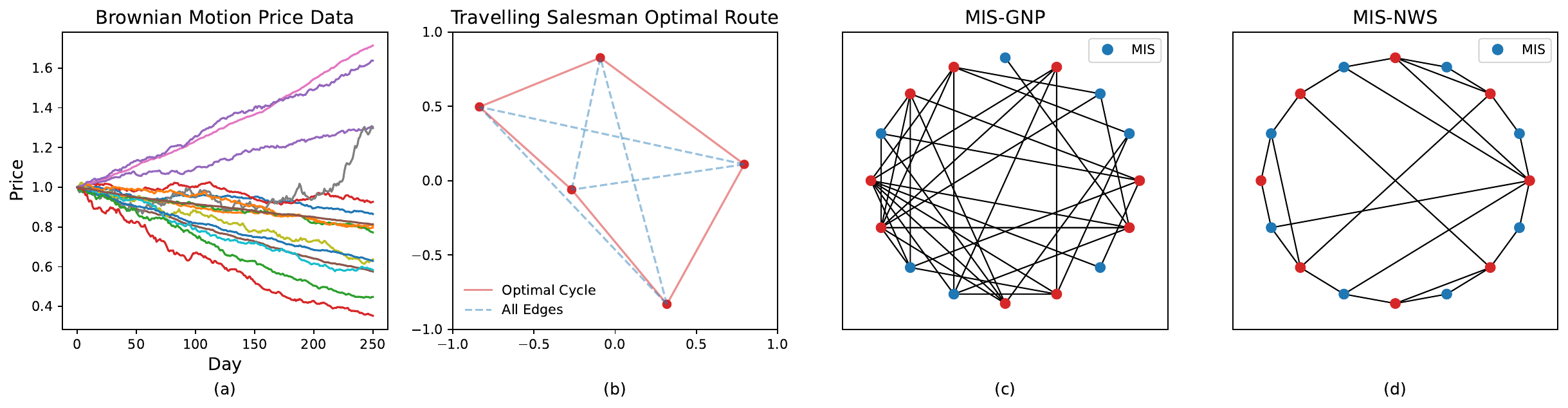}
    \caption{Visualization of example problem instances. From left to right, (a) geometric Brownian prices versus timestep (days), (b) an example Travelling Salesman optimal route, (c) an example Erdős–Rényi graph (max independent set in blue), (d) an example Newman–Watts–Strogatz graph (max independent set in blue)}
    \label{ProVis}
\end{figure*}

Representative quadratic optimization problems were utilized for benchmarking, with most being non-QUBO-Relaxed. These example problems include
\begin{itemize}
    \item Random QUBOS. These are QUBOs where the matrix elements of $Q$ are chosen either uniformly from $[-1,1]$ or discretely from $\{-1,1\}$.
    \item Traveling Salesman Problem (TSP). This problem asks for the route between $5$ points that are randomly placed onto $[-1,1]^2$ which visits each city exactly once while minimizing the distance.
    \item Portfolio Optimization. This problem begins by simulating a set of stock prices over a fixed time duration, and then computing the covariance and mean return of each asset. The optimization then seeks to find the best combination of stocks to buy to simultaneously maximize the return, minimize the risk, and satisfy a predetermined budget constraint. 
    \item Maximum Independent Set (MIS). This problem begins with an unweighted, connected graph, and then aims to determine the largest set of verticies $S$ such that no two elements of $S$ share an edge. We study this problem on two different graph ensembles, the Erdős–Rényi (GNP) model and the Newmann-Watts-Strogatz (NWS) model.
\end{itemize}
Visualizations of these problems are provided in Figure~\ref{ProVis}. A more detailed background of each problems is provided in Appendix~\ref{EXProblem}.
\\\\
To gauge the impact of vertex-at-top rotations, for the $H^\textrm{Adj.}_A$ warmstarts, we tested rotations on the first qubit, the last qubit, and none of the qubits. The choice of rotating the first qubit was chosen arbitrarily: since these problem instances are random, the distribution of the $Q$ matrix elements (of the corresponding QUBO) should be invariant under permutations of rows/columns. However, after transforming the problem into a Max-Cut instance, the last qubit will have a different structure due to it's special role in the mapping in (\ref{QUBO_to_MAXCUT}). The depth-$0$ data in Appendix~\ref{Depth0Data} provides empirical support for this.
\\\\
For all of the following problems, $1000$ instances were generated and ran at depth $p=0$ (see Appendix~\ref{Depth0Data}) and $10$ of those $1000$ were selected to be run at depths $1\leq p \leq 5$. Each problem is a QUBO in $16$ variables with a corresponding $17$-vertex corresponding graph.

\section{Results}
Experimental parameters can be found in Appendix~\ref{Exparams}.
\\\\
For $p=0,1,\ldots,5$, we present the data for the $\textrm{GW}_{2}$ warmstart in Figure \ref{GW2_All}. Full Depth Data for $\textrm{GW}_{3}$, $\textrm{BM}_{2}$, and $\textrm{BM}_{3}$ is provided in Figure~\ref{Full_Comparison_All}. $\textrm{GW}_{2}$ was selected as the representative warmstart because it obtained instance-specific approximation ratios and optimal sampling probabilities which were either comparable to or better than the other warmstarts considered.
\\\\
While both $\textrm{BM}_2$ and $\textrm{BM}_3$ perform comparable to the $\textrm{GW}$ warmstarts in terms of their obtained instance-specific approximation ratio, they obtain significantly lower optimal sampling probabilities. This is likely due to the warmstarts finding relatively high cost solutions, but not necessarily optimal ones.
\\\\
Because the relaxation corresponding to the QUBO-Relaxed warmstart is not convex, there are no guarantees on its performance. We use a LFBBGs optimizer to estimate $y^c$, and vary the number of random initial conditions.\footnote{Portfolio Optimization is actually a QUBO-Relaxed problem when relaxed, so we can solve for $y^c$ directly using QUBO-Relaxed programming. Thus, the number of random initializations does not matter for Portfolio Optimization.} To give a fair comparison between this warmstart and the others, we compare its performance with both $10$ random initializations and $50$ random initializations. As can be seen in the corresponding figures, there is a notable difference in performance between $50$ and $10$ initializations for certain problems.
\\\\
The warmstarts were run with different vertex-at-top rotations, specifically first, last, and no rotation, as discussed in the previous section. A more in depth explanation is provided in Appendix~\ref{Depth0Data}.
\subsection{Aggregate Results}
All trends discussed can be seen in in Figure ~\ref{GW2_All}. The following results are all referring to the average values and standard deviations at $p=5$.
\\\\
We use the results of our experiment to better understand the performance of
\begin{itemize}
    \item \textbf{(1)} GW$_{2}$ based on vertex-at-top rotation choice.
    \item \textbf{(2)} QUBO-Relaxed based on the number of initializations.
    \item \textbf{(3)} GW$_{2}$ Last rotation compared to that of QUBO-Relaxed.
\end{itemize}
\textbf{(1)} 
For both metrics, the vertex on top rotation choices ranked in order of decreasing performance (for all problems except TSP) are last, first, and no rotation. The obtained optimal sampling probabilities/instance-specific approximation ratios for the last and first vertex on top rotations are often within $\pm 0.25$ standard deviations, whereas the obtained optimal sampling probabilities/instance-specific approximation ratios for the first vertex on top rotation and no rotation are typically not within $\pm 0.25$. 
\\\\
These results provide further empirical support that when mapping QUBOs to Max-Cut problems, the additional degree of freedom added has a fundamentally different structure than the other variables, while also supporting the findings in \cite{tate2022bridgingclassicalquantumsdp} where it was shown that vertex-on-top rotations in general provide improved warmstarts. 
\\\\
\textbf{(2)} In general, QUBO-Relaxed with 50 initializations slightly outperforms QUBO-Relaxed with 10 initializations. This is to be expected, because the performance of the non-QUBO-Relaxed optimization is strongly dependent on the number of random initial conditions tested. 
\\\\
Note that Portfolio Optimization is a special case of this warmstart, because the matrix of QUBO coefficients is itself negative semidefinite. As a result, solving the continous relaxation is a QUBO-Relaxed programming problem and can be solved efficiently. Thus, the only difference in obtained optimal sampling probabilities or instance-specific approximation ratios between different optimization runs is due to differences in the QAOA parameter optimization, not the number of initializations.
\\\\
 At depth $p=5$, both random QUBO problem types have average sampling probabilities within 0.25 for 50/10 random initializations.  Interestingly, the difference in obtained optimal sampling probability between 50 and 10 initializations is largest for TSP and the MIS problems, both of which arise from constrained optimization problems. This suggests that constrained optimization problems are intrinsically more difficult to solve, and this structural difference is preserved when these problems are mapped to QAOA.
\\\\
50 and 10 initializations are within 0.25 standard deviations of each other for all problems.
\\\\
\textbf{(3)} In order give a fair comparison between the various GW$_{2}$ vertex on-top rotation choices and different initializations for QUBO-Relaxed, we categorize each vertex on top rotation choice in terms of how it relates to the QUBO-Relaxed warmstart performance (in terms of average values): 
\begin{itemize}
    \item[\textbf{($\uparrow$)}] better than both QUBO-Relaxed with 50 initializations.
    \item[\textbf{\makebox[12pt]{$\hspace{0pt}(\leftrightarrow$)}}] \hspace{2pt}between QUBO-Relaxed with 10 initializations and 50 initializations 
    \item[\textbf{($\downarrow$)}] worse than QUBO-Relaxed with 10 initializations
\end{itemize} 
\begin{figure}[h!]
    \centering
    \resizebox{0.5\textwidth}{!}{
        \begin{tabular}{|p{0.25cm}||p{1.5cm}|p{1.5cm}|p{1.5cm}|p{1.5cm}|} 
        \hline
        & \textbf{\makebox[1.5cm]{\centering Rand.}} & \textbf{\makebox[1.5cm]{\centering TSP}} & \textbf{\makebox[1.5cm]{\centering PO}} & \textbf{\makebox[1.5cm]{\centering MIS}} \\
        \hline
        $\mathcal{P}$ & \textbf{\makebox[1.5cm]{\centering $\downarrow$}} & \textbf{\makebox[1.5cm]{\centering $\uparrow$}} & \textbf{\makebox[1.5cm]{\centering $\uparrow$}} & \textbf{\makebox[1.5cm]{\centering $\leftrightarrow$}} \\ 
        \hline
        $\alpha$ & \textbf{\makebox[1.5cm]{\centering $\uparrow$}} & \textbf{\makebox[1.5cm]{\centering $\downarrow$}} & \textbf{\makebox[1.5cm]{\centering $\uparrow$}} & \textbf{\makebox[1.5cm]{\centering $\leftrightarrow$}} \\
        \hline
        \end{tabular}
    }

    \caption{Classification of best performing vertex on top-rotation $\textrm{GW}_2$ relative to QUBO-relaxed warmstart. For every problem/metric, the last vertex on top rotation was used for comparison except TSP optimal sampling probability where no rotation was used.}
    \label{tab:my_label}
\end{figure}
\noindent
\vspace{-0.5cm}\\
\subsubsection{Random QUBOs}\hfill

\noindent
For both random QUBOs the optimal sampling probability falls under \textbf{($\downarrow$)}, while the instance-specific approximation ratio falls under \textbf{($\uparrow$)}.\\\\
For the instance-specific approximation ratio, $\textrm{GW}_2$ outperforms QUBO-Relaxed with both 50 and 10 initializations, and is above $0.25$ standard deviations for discrete random QUBOs while being within $0.25$ standard deviations for continous random QUBOS. 
\\\\ 
For the optimal sampling probability, QUBO-Relaxed with 10 initializations and 50 initializations are within $0.25$ standard deviations of each other for both continuous and discrete random QUBOs. QUBO-Relaxed with 10 initializations outperforms $\textrm{GW}_2$ in both instances, but it is within $0.25$ of $\textrm{GW}_2$ for continuous random QUBOs whereas it is not for discrete random QUBOs.
\\\\
The difference between these metrics is likely due to $\textrm{GW}_2$ creating initial states that are superpositions of the optimal state and other high-cost but suboptimal states, which carries through subsequent optimizations. 
\\\\
\vspace{-1cm}\\
\subsubsection{Traveling Salesman}\hfill

\noindent
For TSP, the optimal sampling probability falls under \textbf{($\uparrow$)}, while the instance-specific approximation ratio falls under \textbf{($\downarrow$)}.
\\\\
The optimal sampling probability obtained by both GW$_2$ as well as QUBO-Relaxed are relativley low when compared to other problems. This is likely due to the constraints in the cost function being significantly larger (and contributing more terms to the QUBO matrix). For our case of $5$ cities, there are $2^{16}$ possible QUBO bitstrings but only $24$ feasible solutions.
\\\\
With that being said, for the optimal sampling probability, the warmstarts in decreasing order are, GW$_{2}$ no rotation, QUBO-Relaxed with 10 initializations, QUBO-Relaxed with 50. None of these are within $0.25$ standard deviations of each other. Although QUBO-Relaxed performing better with fewer initializations is surprising, it is likely an artifact of the generally low optimal sampling probabilities.
\\\\
For the instance-specific approximation ratio, GW$_2$ is significantly lower than QUBO-Relaxed with 50 initializations, but its within $0.25$ standard deviations for QUBO-Relaxed with 10 initializations. The reason why increasing the number of initializations has a singificant improvement on the cost is likely also a result of the large number of constraint terms in the cost function.
\\\\
\vspace{-1cm}\\
\subsubsection{Portfolio Optimization}\hfill

\noindent
Solving the QUBO-Relaxed warmstart makes portfolio optimization become a QUBO-Relaxed programming problem, and as a result there is no longer a dependence on the number of random initializations. Thus, we don't distinguish between 10 and 50 initializations. With that being said, both metrics fall under \textbf{($\uparrow$)} with $\textrm{GW}_2$ outperforming QUBO-Relaxed.
\\\\
The instance-specific approximation ratios obtained by QUBO-Relaxed and $\textrm{GW}_2$ are both relatively high (when compared to other problems), but $\textrm{GW}_2$ is still more than $0.25$ standard deviations above QUBO-Relaxed. 
\\\\
In terms of optimal sampling probability, $\textrm{GW}_2$ out performs QUBO-Relaxed, but their probabilities are within $0.25$ standard deviations of each other. Notice that these sampling probabilities are significantly further from $1$ than the corresponding instance-specific approximation ratios, corresponding to QAOA states which are superpositions of both the optimal state and other sub-optimal but still high cost sates. For portfolio optimization in particular, because the budget constraint is large, any term that satisfies that constraint will have a high cost by default.
\\\\
The reasoning for why Portfolio Optimization behaves this way is likely due to its QUBO-Relaxedity, which not only changes the optimization routine for the QUBO-Relaxed warmstart but also impacts the performance of $\textrm{GW}_2$ (as can be seen empirically).
\\\\
\vspace{-1cm}\\
\subsubsection{Maximum Independent Set}\hfill

\noindent
Both MIS problems considering both metrics fall under \textbf{($\leftrightarrow$)}, specifically following the order of QUBO-Relaxed with 50 initializations, GW$_2$ Last rotation, then QUBO-Relaxed with 10 initializations.
\\\\
For MIS, the GW$_2$ warmstart is usually within 0.25 standard deviations of QUBO-Relaxed with 50 initializations. The instance-specific approximation ratio for MIS-GNP is the only problem and metric for which QUBO-Relaxed with 50 initializations performs above 0.25 standard deviations of the GW$_2$ warmstart. This suggests that when even GW$_2$ warmstarts preform better than QUBO-Relaxed with 10 but worse than 50 initializations, the margin between the best GW$_2$ warmstart and QUBO-Relaxed is not large.
Furthermore, MIS problems are a constrained optimization task, which might explain why increasing the number of initializations for QUBO-Relaxed provides an advantage over $\textrm{GW}_2$. 
\label{FullDepthData}

\begin{figure*}

    \foreach \i in {1,...,4} {
        \includegraphics[width=0.9\linewidth,page=\i]{Figures/GW2_All.pdf}
        \vspace{-0.3cm}

    }
    \vspace{-.2cm}
    \caption{$(\alpha,\mathcal{P})$ data for $\textrm{GW}_{2}$ over $p$ for the $10$ continuous random QUBO, discrete random QUBO, TSP, Portfolio Optimization, MIS-GNP, and MIS-NWS problem instances. Datapoints are average values and shaded regions are $\pm 0.25$ standard deviations. Only one QUBO-Relaxed warmstart is shown for Portfolio Optimization because solving the warmstart transforms the problem into a QUBO-Relaxed programming problem.}
    \label{GW2_All}
\end{figure*}

\section{Conclusion}
We introduced a new method for applying semidefinite relaxations to solving QUBO problems via QAOA by using a mapping from QUBOs to Max-Cut problems as an intermediate step. We benchmarked this approach on various QUBO problems: Random QUBOs, TSP, Portfolio Optimization, and MIS. As a comparison, we used a non-QUBO-Relaxed warmstart, QUBO-Relaxed. 
\\\\
We found that the best choice of SDP warmstart was Goemans-Williamson projected onto $2$ dimensions, with a rotation applied to the auxiliary variable introduced in the mapping from QUBOs to Max-Cuts.  
\\\\
Because QUBO-Relaxed relies on a non-QUBO-Relaxed optimization problem, its performance is strongly dependent on the number of random initial conditions tested by the optimizer. To give a fair comparison between  $\textrm{GW}_2$  and QUBO-Relaxed we tested both 50 and 10 random initializations for the latter. Empirically we saw that 50 initializations generally outperformed 10 initializations, but this difference was problem dependent. Because this non-QUBO-Relaxed optimization in general has no performance guarantees, we expect that for larger problems, more iterations would be needed.
\\\\
The relative performance of QAOA between $\textrm{GW}_2$ warmstarts and QUBO-Relaxed is problem and metric dependent. Future researchers might be interested in better understanding the mechanisms behind these differences.

\section{Acknowledgements}
This work was supported by the U.S. Department of Energy through the Los Alamos National Laboratory. Los Alamos National Laboratory is operated by Triad National Security, LLC, for the National Nuclear Security Administration of U.S. Department of Energy (Contract No. 89233218CNA000001). The research presented in this article was supported by the Laboratory Directed Research and Development program of Los Alamos National Laboratory under project number 20230049DR as well as by the NNSA's Advanced Simulation and Computing Beyond Moore's Law Program at Los Alamos National Laboratory. Report Number: LA-UR-25-22532.

\bibliographystyle{unsrt}
\bibliography{sample-base}

\clearpage

\appendix
\begin{figure*}
    
    \includegraphics[width=0.8\linewidth]{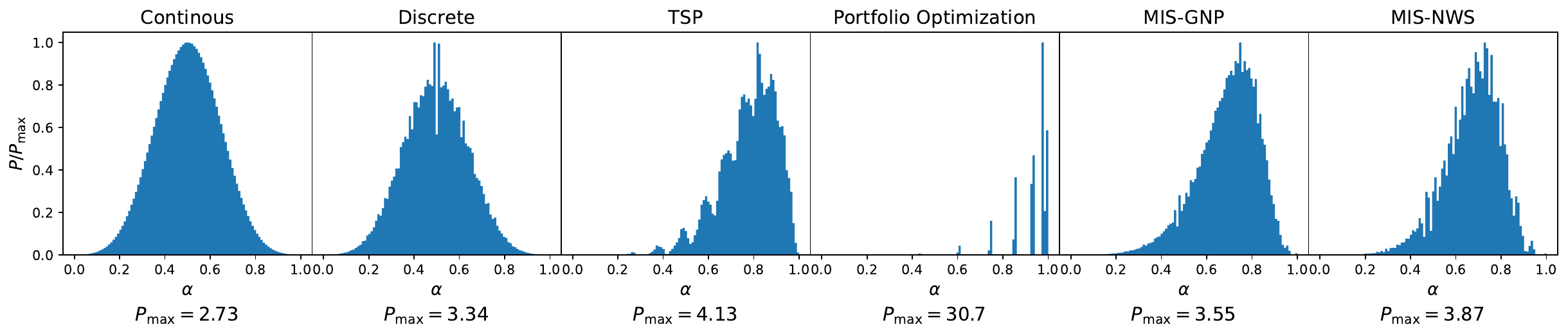}
    \caption{Probability density of $\alpha$ as a function of random states. $y$-axis is the Normalized density so the shape of the distributions can be compared more easily, normalization constants (maximum probability densities, $P_\textrm{max}$) are displayed below. Notice that the constrained optimization problems have skewed distributions.}
    \label{alphadist}
\end{figure*}
\section{QUBO to Ising Mapping Derivation}\label{Map Deriv}
Lemma~\ref{impossiblemap} demonstrates why the extra variable is required when mapping from QUBOs to Ising Problems.
\begin{lemma}\label{impossiblemap}
Given symmetric $Q\in \mathbb{R}^{n\times n}$, there does not exist (in general) some $J\in \mathbb{R}^{n\times n}$ such that for all $y\in\{-1,1\}^n$
\begin{equation}\label{Lemma1eq}
    \sum_{0\leq i,j< n}Q_{i,j}\left(\frac{y_i+1}{2}\right)\left(\frac{y_j+1}{2}\right)\sim\sum_{0\leq i,j< n}J_{i,j}y_iy_j
\end{equation}
Where $\sim$ denotes equality up to a constant difference.
\end{lemma}
\begin{proof}
Assume that such a $J$ exists. Expanding equation (\ref{Lemma1eq}),
\begin{equation}
    \begin{aligned}
        &\sum_{0\le i,j<n}^{ }Q_{i,j}\left(\frac{y_{i}+1}{2}\right)\left(\frac{y_{j}+1}{2}\right)\\
        &=\frac{1}{4}\sum_{0\le i,j<n}^{ }Q_{i,j}\left(1+y_{i}+y_{j}+y_{i}y_{j}\right)\\
        &=\frac{1}{4}\textrm{sum}(Q)+\frac{1}{4}\sum_{0\le i,j<n}^{ }Q_{i,j}y_{i}y_{j}+\frac{1}{2}\sum_{0\le i,j<n}^{ }Q_{i,j}y_{i}
    \end{aligned}
\end{equation}
Thus, the left side of (\ref{Lemma1eq}) has terms which are linear in $y$ whereas the right side only consists of quadratic terms. Thus, no such $J$ could exist.
\end{proof}\noindent
As a result of Lemma~\ref{impossiblemap}, its clear that there is no simple mapping from a $n$-variable QUBO to an $n$-variable Ising problem. However, if we introduce degeneracy into the latter, then it becomes possible to construct such a mapping.
\\\\
The degeneracy will rise from adding an extra degree of freedom, represented by $y_n$, to the Ising problem and changing the mapping to \footnote{Note that both $y$ and $-y$ get mapped to the same $x$. Hence this }
$$x_i = \frac{1}{2}(1-y_iy_n)$$
\begin{lemma}\label{impossiblemap}
Given symmetric $Q\in \mathbb{R}^{n\times n}$, there is an unique $J\in \mathbb{R}^{(n+1)\times (n+1)}$ and constant $c\in R$ such that for all $y\in\{-1,1\}^{n+1}$
\begin{equation}\label{Lemma2eq}
\hspace{-0.5cm}
    \sum_{0\le i,j\le n}^{ }J_{i,j}y_{i}y_{j}+c=\frac{1}{4}\sum_{0\le i,j<n}^{ }Q_{i,j}\left(1-y_{i}y_{n}\right)\left(1-y_{j}y_{n}\right)
\end{equation}
\end{lemma}
\begin{proof}
Expanding equation (\ref{Lemma2eq}) by seperating the terms out that depend on $y_n$ gives,
\begin{equation}
    \begin{aligned}
        J_{n,n}y_{n,n}+2\sum_{0\le i<n}^{ }J_{i,j}y_{i}y_{n}+\sum_{0\le i,j<n}^{ }J_{i,j}y_{i}y_{j}+c\\=\frac{1}{4}\textrm{sum(Q)}+\frac{1}{4}\sum_{0\le i,j<n}^{ }Q_{i,j}y_{i}y_{j}-\frac{1}{2}\sum_{0\le i,j<n}^{ }Q_{i,j}y_{i}y_{n},
    \end{aligned}
\end{equation}
from which it clearly follows that
\begin{equation}
    \begin{aligned}
    J_{i,j}&=\frac{1}{4}Q_{i,j}\\
J_{n,i}&=-\frac{1}{4}\sum_{0\leq j < n}Q_{i,j}\\
J_{n,n}&=0\\
c&=\frac{1}{4}\textrm{sum}(Q).
 \end{aligned}
\end{equation}
\end{proof}
\noindent
Equation (\ref{QUBO_to_MAXCUT}) is a trivial corollary of this lemma.
\section{Problem Instances Background}\label{EXProblem}

\subsection{Random QUBOs}
Random QUBO problem instances are obtained by randomly generating symmetric matrices $Q\in\mathbb{R}^{16\times 16}$. The elements of the upper-triangle of $Q$ (i.e. $Q_{ij}$ with $i \leq j$) are decided by taking independent samples from a fixed distribution; we consider two choices for the distribution below:
\begin{itemize}
    \item the (continuous) uniform distribution on $[-1,1]$, and 
    \item the (discrete) uniform distribution on $\{-1,1\}$.
\end{itemize}
The remaining elements of each matrix $Q$ (i.e. $Q_{ij}$ with $i > j$) are then uniquely determined due to the requirement that $Q$ be symmetric.

\subsection{Travelling Salesman}

Given a graph $G=(V,E)$ with adjacency matrix $A$, the Travelling Salesman Problem (TSP) \cite{ref1} is finding the Hamiltonian cycle of $G$ with the smallest total edge weight. Because the Hamiltonian Cycle problem is itself NP-Hard \cite{Book_1975}, the TSP problem is NP-Hard as well. Expansions of TSP have applications to various applications such as Vehicle Routing \cite{Dorling_2017}, Disaster Retrieval \cite{maziero2021branchandcutalgorithmscoveringsalesman}, and Equitable Routing \cite{bhadoriya2024equitableroutingrethinking}.
\\\\
Following \cite{Qian_2023}, we focus on the case of fully connected symmetric TSP, corresponding to complete undirected graphs.  A given sequence of vertices can be encoding using binary decision variables $x\in\{0,1\}^{|V|^2}$, where (for $0\leq i,t<|V|$ ) $x_{|V|t+i}\defeq x_{t,i}=1$ if and only if vertex $i$ is visited at time $t$. The optimal sequence of vertices for a given TSP problem corresponds to the $x$ which minimizes\footnote{Here $x_{|V|,i}=x_{0,i}$ i.e. periodicity is enforced by construction.}
\begin{equation}
\begin{aligned}
C(x)&=C_{\textrm{dist}}(x)+\lambda{C_{\textrm{penalty}}(x)}\\
&=\sum_{0\leq i,j < |V|}A_{i,j}\sum_{0\leq t <|V|}x_{t,i}x_{(t+1),j}\\&+\lambda\sum_{0\leq t < |V|}\left(1-\sum_{0\leq i < |V|}^{n-1}x_{t,i}\right)^2\\&+\lambda
\sum_{0\leq i < |V|}\left(1-\sum_{0\leq t < |V|}x_{t,i}\right)^2
\end{aligned}
\end{equation}
Penalty terms are included to ensure that the vertex sequence is indeed a Hamiltonian cycle (i.e. that every vertex is visited exactly once and there is exactly one vertex for each timestep). The constant $\lambda\in \mathbb{R}$ is a Lagrange multiplier. To ensure that the Hamiltonian cycle constraints are satisfied we require $\lambda>\max(A_{i,j})$.
\\\\
The cost function can be further improved by eliminating rotational symmetry. To do so, we fix $x_{0}=1$ and define $\tilde x \in \{0,1\}^{(|V|-1)^2}$ so that (for $0\leq i,t<|V|-1$)
\begin{equation}
    \tilde x_{(|V|-1)t+i}\defeq \tilde x_{t,i}=x_{(t+1),(i+1)}.
\end{equation}
Re-expressing $C$ in terms of $\tilde x$
\begin{equation}
\begin{aligned}
C(x)=\tilde C(\tilde x)&=\tilde C_{\textrm{dist}}(\tilde x)+\lambda{\tilde C_{\textrm{penalty}}(\tilde x)}\\& + \sum_{0\leq i < |V|-1}A_{0,i}(\tilde x_{0,i}+\tilde x_{(|V|-2),i})\\
&=\sum_{0\leq i,j < |V|-1}A_{i,j}\sum_{0\leq t <|V|-1}\tilde x_{t,i} \tilde x_{(t+1),j}\\&+\lambda\sum_{0\leq t < |V|-1}\left(1-\sum_{0\leq i < |V|-1}^{n-1}\tilde x_{t,i}\right)^2\\&+\lambda
\sum_{0\leq i < |V|-1}\left(1-\sum_{0\leq t < |V|-1}\tilde x_{t,i}\right)^2.
\end{aligned}
\end{equation}
Because this function is quadratic in the entries of $\tilde x$, we can define a a QUBO\footnote{The minus sign comes from the fact that we defined QUBOs as maximization problems.}with $Q\in\mathbb{R}^{(|V|-1)^2\times (|V|-1)^2}$
\begin{equation}
    \tilde x^T Q\tilde x = -\tilde C(\tilde x)
\end{equation}
To generate problem instances, we sampled points $\{p_i\}_{0\leq i < 5}$ uniformly from $[-1,1]^2$ and  set $A_{i,j}$ in the adjacency matrix $A$ to be the Euclidean distance from $p_i$ to $p_j$, i.e., $A_{i,j}=\|p_i-p_j\|_2$. The Lagrange multiplier $\lambda$ was set to $\lambda = 1.1\max(A_{i,j})$.

\subsection{Portfolio Optimization}
Portfolio Optimization is the process of diversifying one's assets to ensure an appropriate balance between risk and expected returns. Portfolio Optimization has been used as a benchmarking problem for different variations of QAOA \cite{Egger_2021, Brandhofer_2022}. 

Given the asset price of $n$ stocks, the goal is to find a vector $x\in \{0,1\}^n$ that maximizes return, minimizes risk, and satisfies $1^Tx=B$ where $B$ is a budget $1\leq B \leq n$.
\\\noindent
Following \cite{Egger_2021}, $S_{i,k}$, the price of asset $i$ over time $k$ is generated using Geometric Brownian Motion for $N=250$ time steps and $n$ assets:
\begin{equation}
    S_{i,k}=S_{i,0}\exp[(\mu_{i}-\sigma_{i}^2/2)k/N+\sigma_{i}W_{k}],
\end{equation}
with $S_{i,0}=1$ for all $0\leq i < n$. Both the drifts, $\mu_{i}$, and the volatilities $\sigma_{i}$ are randomly generated using uniform distribution set within the range $[-0.05,0.05]$ and $[-0.20,0.20]$, respectively. $W_{k} = \sum_{l=0}^jz_{l}/\sqrt{N}$ represents the cumulative Brownian motion, where $z_{l}$ is a random variable with a standard normal distribution.
\\
For a set of asset prices, the return of asset $i$ from time $k$ to $k+1$ is, $r_{i,k} = S_{i,k}/S_{i,k-1}-1$. Using these returns, a covariance matrix $\Sigma\in\mathbb{R}^{n\times n}$ and a mean return vector $\mu\in\mathbb{R}^n$ can be calculated \cite{Buonaiuto2023}.
\\
The optimal asset vector $x$ will then maximize
\begin{equation}
    \mu^T x -qx^T\Sigma x -\lambda(1^Tx-B)^2,
\end{equation}
where $\lambda$ is a Lagrange multiplier for the penalty (budget) constraint.
\\\\
This problem can be converted to a QUBO with $Q\in \mathbb{R}^{n\times n}$ such that
\begin{equation}
    x^TQx =\mu^T x -qx^T\Sigma x -\lambda(1^Tx-B)^2.
\end{equation}
For our experiments, we set $\lambda = \textrm{sum}(|\Sigma|)+\textrm{sum}(|\mu|)$, $n=16$, $B=8$, $q=0.5$.

\subsection{Maximum Independent Set}
Given a graph $G=(V,E)$ a subset $U \subseteq V$ is said to be an independent set if there are no edges between any of the vertices of $U$.
\\\\
The Maximum Independent Set (MIS) problem is to find an independent set of maximum cardinality of an arbitrary graph.
This problem is equivalent to the NP-hard set packing problem \cite{Book_1975}. There exist greedy approximate (classical) algorithms for MIS \cite{Coja_Oghlan_2014}. The MIS problem has also been previously studied with QAOA in \cite{farhi2020quantumapproximateoptimizationalgorithm,9550042}.
Given a weighted graph $G = (V,E)$ with edge weights $w: E \to \mathbb{R}$, the MIS can be formulated as a QUBO \cite{Lucas_2014} where $Q\in\mathbb{R}^{|V|\times |V|}$
\begin{equation}
    x^TQx=\sum_{0\leq i < |V|}x_i-c\sum_{(i,j) \in E}w_{ij} x_ix_j
\end{equation}
and $c$ is a free variable controlling the weight of the penalty term. To ensure that the optimal vertex set is indeed independent it's required that $c>1$. For our simulations, $c$ was set to $1.1$.
\\\\
The graphs
were sampled from two distributions:
\footnote{The values for graph distrubution parameters were chosen so that both the binomial and Newman-Watts-Strogatz graphs had similar independence numbers on average. If the graph generated is not connected we resample the distribution (i.e. rejection sample).}
\begin{itemize}
    \item Erdős–Rényi random graphs (GNP) (see \cite{Frieze_Karoński_2015}). Here $n=16,p=0.25$,
    \item Newman–Watts–Strogatz random graphs (NWS) (see \cite{Watts1998}). Here, $n=16,k=3,p=0.5$.
\end{itemize}

\section{Depth 0 Data}\label{Depth0Data}
Given a warmstart and a corresponding $|\psi_\textrm{init}\rangle$, at depth $0$ we have (for $H$, $E_\textrm{min}$, $E_\textrm{max}$),
\begin{equation}
    \begin{aligned}
        \alpha &= \frac{\langle \psi_\textrm{init} | H | \psi_\textrm{init}\rangle-E_\textrm{min}}{E_\textrm{max}-E_\textrm{min}}\\
\mathcal{P} &= \sum_{|\phi\rangle}|\langle \phi | \psi_\textrm{init}\rangle|^2 \\ & \textrm{where} \quad \langle \phi | H| \phi\rangle = E_\textrm{max}.
    \end{aligned}
\end{equation}
Because $\alpha$ for the optimal depth-$p$ parameters can only increase with $p$, depth-$0$ results are a useful heuristic for estimating the relative performance of warmstart techniques \cite{tate2022bridgingclassicalquantumsdp}.
In particular, we are interested in how the choice of vertex-at-top rotation qubit affects the obtained $(\alpha,\mathcal{P})$. A natural assumption would be that there are 3 unique distributions $(\alpha,\mathcal{P})$ dependent on the choice of vertex-at-top rotation: 
\begin{enumerate}
    \item vertex-at-top rotation on any of the qubits except the last qubit,
    \item vertex-at-top rotation on the last qubit,
    \item and no vertex-at-top rotation.
\end{enumerate}
Figure~\ref{HistData} contains a datatable with the frequency at which each choice of vertex at top rotation maximizes either $\alpha$ or $\mathcal{P}$. Outliers are indicated, and notably only appear in the data for vertex-at-top rotations applied to the last qubit or when no vertex-at-top rotation is applied. 
\\\\
Based on this assumption, it is sufficient to study $3$ vertex-at-top rotation choices, with each each corresponding to one of these $3$ distributions. As in Section~\ref{FullDepthData}, we consider vertex-at-top rotations on the First, Last, and None of the qubits.\\ \\Figure~\ref{CompData} contains a datatable with the mean and standard deviation for $\alpha,\mathcal{P}$ for each warmstart with each of these vertex-at-top rotation options.

\section{\texorpdfstring{$\alpha$ }
{TEXT}  Distributions}\label{alphadistributions}
 The instance-specific approximation ratio ($\alpha$) is a convenient metric because it is normalized to $[0,1]$ for all problems. However, the distribution of $\alpha$ over states is problem dependent.

 More formally, we want to consider the continuous random variable $\boldsymbol\alpha \sim P(\alpha)$ given by
\begin{equation}
    \boldsymbol\alpha = \frac{X^T\mathbf{A}X-\min_{x\in \{-1,1\}^{n}}x^T\mathbf{A}x}{\max_{x\in \{-1,1\}^{n}}x^T\mathbf{A}x-\min_{x\in \{-1,1\}^{n}}x^T\mathbf{A}x},
\end{equation}
where (for a given problem type) $\mathbf{A}$ is a random matrix sampled by obtaining the adjacency matrix of the corresponding graph of a randomly generated problem instance, and $X$ is a random element of $\{-1,1\}^n$ with uniform distribution.
\\\\
To quantify the dependence of the distribution on problem type, we estimated the probability distribution  $P(\alpha)$ for each problem (Figure~\ref{alphadist}). For each problem instance, $\alpha$ was computed for each of the $2^{17}$ computational basis states, yielding a total of $1000\cdot2^{17}$ samples for each problem type. The unit interval was then divided into $101$ uniformly spaced sub-intervals. The probability density of $\alpha$ was approximated to be piecewise constant on each of these sub-intervals, with value equal to the number of $\alpha$ samples in the interval divided by the $1000/101\cdot 2^{17}$ (the product of the width of the subinterval-intervals with the total number of samples).

The most notable feature of these distributions (besides their problem dependence) is that the distribution for the Random QUBOs is symmetric, whereas the distributions for TSP, Portfolio Optimization, and MIS are skewed left. This is likely due to the constraint terms in the latter, which results in states satisfying the constraints having comparatively large $\alpha$ values, even if they are not optimal.
\section{Experimental Parameters}\label{Exparams}
All code used to generate data is available at \cite{Bhattacharya_QUBO_to_MaxCut}.
\begin{center}
Warmstart info\vspace{0.1cm}
\fontsize{8pt}{8pt}\selectfont

  \begin{tabular}{p{0.14\textwidth}p{.14\textwidth} p{0.14\textwidth}} 

 \hline\\
 $\hspace{0.025\textwidth}\textrm{QUBO-Relaxed}$ & $\hspace{0.07\textwidth}\textrm{BM}_k$ & $\hspace{0.07\textwidth}\textrm{GW}_k$ \\
 \hline\hline 

     \begin{itemize}
         \item $\varepsilon=0.1$
         \item $y^TQy$ was optimized using L-BFGS-B algorithm.
     \end{itemize}
  &  
     \begin{itemize}
         \item 100iterations
         \item 50 initial conditions
         \item $\eta=0.05$
     \end{itemize}
  &  

     \begin{itemize}
         \item 50\hspace{0.1cm}random bases sampled
     \end{itemize}\\

 \hline
\end{tabular}
\end{center}
All QAOA runs were done as follows:
\begin{itemize}
    \item The parameters were optimized using COBYLA.    \item For each circuit, the QAOA optimization loop was ran 10 times, each time with a different starting initialization of the parameters. If $p=1$, these were drawn uniformly. If $p>1$, $9$ of the 10 initializations were drawn uniformly and the other was the best performing (cost wise) initial parameters from depth $p-1$.
    \item The final output of a circuit was the optimized parameters (out of the set of $10$) with the largest cost.
    \item All simulations were computed with a custom QAOA simulator based on \cite{Lykov_2023}.
\end{itemize}
 
\begin{figure*}

\begin{center}
\captionof*{table}{\fontsize{8pt}{8pt}\selectfont Continous}
\vspace{-0.3cm}

\resizebox{\textwidth}{!}{%
\begin{tabular}{|l||*{18}{c|}}\hline 
 \backslashbox{$\alpha$}{$\mathcal{P}$}&0&1&2&3&4&5&6&7&8&9&10&11&12&13&14&15&16&None
\\\hline\hline
$\textrm{BM}_2$ & \backslashbox{ 58 }{ 42 }&\backslashbox{ 37 }{ 57 }&\backslashbox{ 46 }{ 58 }&\backslashbox{ 51 }{ 46 }&\backslashbox{ 41 }{ 64 }&\backslashbox{ 46 }{ 51 }&\backslashbox{ 44 }{ 51 }&\backslashbox{ 51 }{ 56 }&\backslashbox{ 44 }{ 62 }&\backslashbox{ 56 }{ 46 }&\backslashbox{ 51 }{ 54 }&\backslashbox{ 47 }{ 54 }&\backslashbox{ 58 }{ 50 }&\backslashbox{ 47 }{ 50 }&\backslashbox{ 44 }{ 51 }&\backslashbox{ 51 }{ 60 }&\backslashbox{ \color{red}{\textbf{195}} }{ \color{blue}{\textbf{90}} }&\backslashbox{ 33 }{ 58 }\\\hline 
$\textrm{BM}_3$ & \backslashbox{ 43 }{ 52 }&\backslashbox{ 40 }{ 64 }&\backslashbox{ 38 }{ 63 }&\backslashbox{ 49 }{ 60 }&\backslashbox{ 44 }{ 50 }&\backslashbox{ 46 }{ 61 }&\backslashbox{ 40 }{ 54 }&\backslashbox{ 48 }{ 60 }&\backslashbox{ 46 }{ 48 }&\backslashbox{ 47 }{ 68 }&\backslashbox{ 53 }{ 55 }&\backslashbox{ 45 }{ 47 }&\backslashbox{ 39 }{ 45 }&\backslashbox{ 40 }{ 57 }&\backslashbox{ 34 }{ 49 }&\backslashbox{ 41 }{ 59 }&\backslashbox{ \color{red}{\textbf{221}} }{ \color{blue}{\textbf{82}} }&\backslashbox{ 86 }{ 26 }\\\hline 
$\textrm{GW}_2$ & \backslashbox{ 45 }{ 49 }&\backslashbox{ 41 }{ 64 }&\backslashbox{ 42 }{ 55 }&\backslashbox{ 51 }{ 64 }&\backslashbox{ 40 }{ 48 }&\backslashbox{ 56 }{ 66 }&\backslashbox{ 49 }{ 45 }&\backslashbox{ 48 }{ 53 }&\backslashbox{ 47 }{ 69 }&\backslashbox{ 44 }{ 57 }&\backslashbox{ 52 }{ 53 }&\backslashbox{ 45 }{ 46 }&\backslashbox{ 51 }{ 66 }&\backslashbox{ 47 }{ 48 }&\backslashbox{ 48 }{ 66 }&\backslashbox{ 57 }{ 51 }&\backslashbox{ \color{red}{\textbf{227}} }{ 76 }&\backslashbox{ 14 }{ 28 }\\\hline 
$\textrm{GW}_3$ & \backslashbox{ 43 }{ 49 }&\backslashbox{ 29 }{ 60 }&\backslashbox{ 37 }{ 65 }&\backslashbox{ 44 }{ 57 }&\backslashbox{ 39 }{ 52 }&\backslashbox{ 57 }{ 54 }&\backslashbox{ 44 }{ 48 }&\backslashbox{ 41 }{ 52 }&\backslashbox{ 45 }{ 64 }&\backslashbox{ 50 }{ 60 }&\backslashbox{ 43 }{ 55 }&\backslashbox{ 43 }{ 50 }&\backslashbox{ 50 }{ 71 }&\backslashbox{ 51 }{ 51 }&\backslashbox{ 44 }{ 59 }&\backslashbox{ 52 }{ 70 }&\backslashbox{ \color{red}{\textbf{287}} }{ 75 }&\backslashbox{ 1 }{ 8 }\\\hline 
\end{tabular}}
\end{center}

\begin{center}
\captionof*{table}{\fontsize{8pt}{8pt}\selectfont Discrete}
\resizebox{\textwidth}{!}{%
\begin{tabular}{|l||*{18}{c|}}\hline 
 \backslashbox{$\alpha$}{$\mathcal{P}$}&0&1&2&3&4&5&6&7&8&9&10&11&12&13&14&15&16&None
\\\hline\hline
$\textrm{BM}_2$ & \backslashbox{ 56 }{ 59 }&\backslashbox{ 44 }{ 51 }&\backslashbox{ 50 }{ 55 }&\backslashbox{ 49 }{ 42 }&\backslashbox{ 61 }{ 52 }&\backslashbox{ 41 }{ 56 }&\backslashbox{ 33 }{ 51 }&\backslashbox{ 45 }{ 40 }&\backslashbox{ 44 }{ 51 }&\backslashbox{ 48 }{ 55 }&\backslashbox{ 65 }{ 67 }&\backslashbox{ 55 }{ 37 }&\backslashbox{ 52 }{ 64 }&\backslashbox{ 53 }{ 55 }&\backslashbox{ 33 }{ 50 }&\backslashbox{ 54 }{ 48 }&\backslashbox{ \color{red}{\textbf{174}} }{ \color{blue}{\textbf{106}} }&\backslashbox{ 43 }{ 61 }\\\hline 
$\textrm{BM}_3$ & \backslashbox{ 37 }{ 57 }&\backslashbox{ 49 }{ 60 }&\backslashbox{ 38 }{ 61 }&\backslashbox{ 39 }{ 47 }&\backslashbox{ 33 }{ 50 }&\backslashbox{ 46 }{ 65 }&\backslashbox{ 53 }{ 44 }&\backslashbox{ 41 }{ 61 }&\backslashbox{ 38 }{ 35 }&\backslashbox{ 37 }{ 69 }&\backslashbox{ 27 }{ 58 }&\backslashbox{ 50 }{ 59 }&\backslashbox{ 54 }{ 55 }&\backslashbox{ 46 }{ 59 }&\backslashbox{ 47 }{ 58 }&\backslashbox{ 35 }{ 57 }&\backslashbox{ \color{red}{\textbf{244}} }{ 75 }&\backslashbox{ 86 }{ 30 }\\\hline 
$\textrm{GW}_2$ & \backslashbox{ 54 }{ 64 }&\backslashbox{ 46 }{ 55 }&\backslashbox{ 49 }{ 55 }&\backslashbox{ 51 }{ 52 }&\backslashbox{ 46 }{ 54 }&\backslashbox{ 55 }{ 62 }&\backslashbox{ 42 }{ 58 }&\backslashbox{ 46 }{ 56 }&\backslashbox{ 37 }{ 37 }&\backslashbox{ 58 }{ 60 }&\backslashbox{ 50 }{ 60 }&\backslashbox{ 44 }{ 66 }&\backslashbox{ 54 }{ 40 }&\backslashbox{ 50 }{ 47 }&\backslashbox{ 57 }{ 72 }&\backslashbox{ 42 }{ 53 }&\backslashbox{ \color{red}{\textbf{204}} }{ \color{blue}{\textbf{88}} }&\backslashbox{ 17 }{ 26 }\\\hline 
$\textrm{GW}_3$ & \backslashbox{ 38 }{ 66 }&\backslashbox{ 39 }{ 62 }&\backslashbox{ 56 }{ 55 }&\backslashbox{ 45 }{ 48 }&\backslashbox{ 41 }{ 58 }&\backslashbox{ 51 }{ 64 }&\backslashbox{ 37 }{ 49 }&\backslashbox{ 44 }{ 65 }&\backslashbox{ 31 }{ 45 }&\backslashbox{ 57 }{ 64 }&\backslashbox{ 50 }{ 63 }&\backslashbox{ 51 }{ 59 }&\backslashbox{ 47 }{ 44 }&\backslashbox{ 45 }{ 50 }&\backslashbox{ 48 }{ 64 }&\backslashbox{ 44 }{ 55 }&\backslashbox{ \color{red}{\textbf{276}} }{ 85 }&\backslashbox{ 0 }{ 5 }\\\hline 
\end{tabular}}
\end{center}

\begin{center}
\captionof*{table}{\fontsize{8pt}{8pt}\selectfont TSP}
\resizebox{\textwidth}{!}{%
\begin{tabular}{|l||*{18}{c|}}\hline 
 \backslashbox{$\alpha$}{$\mathcal{P}$}&0&1&2&3&4&5&6&7&8&9&10&11&12&13&14&15&16&None
\\\hline\hline
$\textrm{BM}_2$ & \backslashbox{ 15 }{ 59 }&\backslashbox{ 25 }{ 46 }&\backslashbox{ 16 }{ 53 }&\backslashbox{ 20 }{ 50 }&\backslashbox{ 26 }{ 50 }&\backslashbox{ 22 }{ 44 }&\backslashbox{ 23 }{ 50 }&\backslashbox{ 19 }{ 50 }&\backslashbox{ 19 }{ 62 }&\backslashbox{ 27 }{ 49 }&\backslashbox{ 19 }{ 46 }&\backslashbox{ 31 }{ 53 }&\backslashbox{ 20 }{ 71 }&\backslashbox{ 29 }{ 53 }&\backslashbox{ 29 }{ 51 }&\backslashbox{ 22 }{ 59 }&\backslashbox{ \color{red}{\textbf{620}} }{ 65 }&\backslashbox{ 18 }{ \color{blue}{\textbf{89}} }\\\hline 
$\textrm{BM}_3$ & \backslashbox{ 8 }{ 45 }&\backslashbox{ 5 }{ 53 }&\backslashbox{ 15 }{ 60 }&\backslashbox{ 10 }{ 68 }&\backslashbox{ 9 }{ 59 }&\backslashbox{ 9 }{ 77 }&\backslashbox{ 6 }{ 57 }&\backslashbox{ 13 }{ 58 }&\backslashbox{ 13 }{ 61 }&\backslashbox{ 5 }{ 69 }&\backslashbox{ 6 }{ 52 }&\backslashbox{ 8 }{ 57 }&\backslashbox{ 11 }{ 55 }&\backslashbox{ 13 }{ 56 }&\backslashbox{ 15 }{ 60 }&\backslashbox{ 9 }{ 58 }&\backslashbox{ \color{red}{\textbf{795}} }{ 40 }&\backslashbox{ 50 }{ 15 }\\\hline 
$\textrm{GW}_2$ & \backslashbox{ 2 }{ 35 }&\backslashbox{ 1 }{ 39 }&\backslashbox{ 1 }{ 27 }&\backslashbox{ 0 }{ 33 }&\backslashbox{ 1 }{ 50 }&\backslashbox{ 0 }{ 40 }&\backslashbox{ 0 }{ 54 }&\backslashbox{ 0 }{ 63 }&\backslashbox{ 0 }{ 51 }&\backslashbox{ 0 }{ 50 }&\backslashbox{ 1 }{ 52 }&\backslashbox{ 0 }{ 55 }&\backslashbox{ 0 }{ 29 }&\backslashbox{ 1 }{ 33 }&\backslashbox{ 0 }{ 33 }&\backslashbox{ 1 }{ 30 }&\backslashbox{ \color{red}{\textbf{985}} }{ \color{blue}{\textbf{282}} }&\backslashbox{ 7 }{ 44 }\\\hline 
$\textrm{GW}_3$ & \backslashbox{ 0 }{ 38 }&\backslashbox{ 0 }{ 42 }&\backslashbox{ 0 }{ 32 }&\backslashbox{ 0 }{ 36 }&\backslashbox{ 0 }{ 46 }&\backslashbox{ 0 }{ 49 }&\backslashbox{ 0 }{ 59 }&\backslashbox{ 0 }{ 58 }&\backslashbox{ 0 }{ 60 }&\backslashbox{ 0 }{ 47 }&\backslashbox{ 0 }{ 45 }&\backslashbox{ 0 }{ 64 }&\backslashbox{ 0 }{ 28 }&\backslashbox{ 0 }{ 42 }&\backslashbox{ 0 }{ 34 }&\backslashbox{ 0 }{ 30 }&\backslashbox{ \color{red}{\textbf{1000}} }{ \color{blue}{\textbf{273}} }&\backslashbox{ 0 }{ 17 }\\\hline 
\end{tabular}}
\end{center}

\begin{center}
\captionof*{table}{\fontsize{8pt}{8pt}\selectfont Portfolio Optimization}
\resizebox{\textwidth}{!}{%
\begin{tabular}{|l||*{18}{c|}}\hline 
 \backslashbox{$\alpha$}{$\mathcal{P}$}&0&1&2&3&4&5&6&7&8&9&10&11&12&13&14&15&16&None
\\\hline\hline
$\textrm{BM}_2$ & \backslashbox{ 51 }{ 66 }&\backslashbox{ 66 }{ 54 }&\backslashbox{ 48 }{ 45 }&\backslashbox{ 49 }{ 44 }&\backslashbox{ 55 }{ 58 }&\backslashbox{ 56 }{ 61 }&\backslashbox{ 56 }{ 45 }&\backslashbox{ 49 }{ 65 }&\backslashbox{ 61 }{ 58 }&\backslashbox{ 62 }{ 36 }&\backslashbox{ 60 }{ 49 }&\backslashbox{ 66 }{ 68 }&\backslashbox{ 51 }{ 38 }&\backslashbox{ 49 }{ 42 }&\backslashbox{ 59 }{ 61 }&\backslashbox{ 53 }{ 56 }&\backslashbox{ \color{red}{\textbf{76}} }{ \color{blue}{\textbf{104}} }&\backslashbox{ 33 }{ 50 }\\\hline 
$\textrm{BM}_3$ & \backslashbox{ 50 }{ 59 }&\backslashbox{ 48 }{ 49 }&\backslashbox{ 60 }{ 55 }&\backslashbox{ 61 }{ 50 }&\backslashbox{ 41 }{ 60 }&\backslashbox{ 41 }{ 58 }&\backslashbox{ 57 }{ 57 }&\backslashbox{ 45 }{ 49 }&\backslashbox{ 47 }{ 63 }&\backslashbox{ 64 }{ 44 }&\backslashbox{ 49 }{ 62 }&\backslashbox{ 33 }{ 64 }&\backslashbox{ 52 }{ 46 }&\backslashbox{ 53 }{ 43 }&\backslashbox{ 67 }{ 52 }&\backslashbox{ 62 }{ 62 }&\backslashbox{ 61 }{ \color{blue}{\textbf{111}} }&\backslashbox{ \color{red}{\textbf{109}} }{ 16 }\\\hline 
$\textrm{GW}_2$ & \backslashbox{ 88 }{ 91 }&\backslashbox{ 88 }{ 88 }&\backslashbox{ 89 }{ 90 }&\backslashbox{ 90 }{ 94 }&\backslashbox{ 86 }{ 84 }&\backslashbox{ 79 }{ 88 }&\backslashbox{ 90 }{ 96 }&\backslashbox{ 91 }{ 90 }&\backslashbox{ 86 }{ 86 }&\backslashbox{ 87 }{ 92 }&\backslashbox{ 84 }{ 79 }&\backslashbox{ 72 }{ 71 }&\backslashbox{ 80 }{ 79 }&\backslashbox{ 101 }{ 98 }&\backslashbox{ 79 }{ 72 }&\backslashbox{ 82 }{ 78 }&\backslashbox{ 110 }{ 106 }&\backslashbox{ 6 }{ 22 }\\\hline 
$\textrm{GW}_3$ & \backslashbox{ 48 }{ 52 }&\backslashbox{ 49 }{ 48 }&\backslashbox{ 61 }{ 69 }&\backslashbox{ 61 }{ 60 }&\backslashbox{ 66 }{ 67 }&\backslashbox{ 56 }{ 61 }&\backslashbox{ 53 }{ 58 }&\backslashbox{ 66 }{ 63 }&\backslashbox{ 51 }{ 47 }&\backslashbox{ 57 }{ 68 }&\backslashbox{ 62 }{ 62 }&\backslashbox{ 41 }{ 44 }&\backslashbox{ 63 }{ 64 }&\backslashbox{ 65 }{ 57 }&\backslashbox{ 44 }{ 46 }&\backslashbox{ 55 }{ 56 }&\backslashbox{ \color{red}{\textbf{114}} }{ \color{blue}{\textbf{122}} }&\backslashbox{ 0 }{ 9 }\\\hline 
\end{tabular}}
\end{center}

\begin{center}
\captionof*{table}{\fontsize{8pt}{8pt}\selectfont MIS-GNP}
\resizebox{\textwidth}{!}{%
\begin{tabular}{|l||*{18}{c|}}\hline 
 \backslashbox{$\alpha$}{$\mathcal{P}$}&0&1&2&3&4&5&6&7&8&9&10&11&12&13&14&15&16&None
\\\hline\hline
$\textrm{BM}_2$ & \backslashbox{ 42 }{ 60 }&\backslashbox{ 41 }{ 47 }&\backslashbox{ 48 }{ 61 }&\backslashbox{ 51 }{ 46 }&\backslashbox{ 40 }{ 58 }&\backslashbox{ 49 }{ 56 }&\backslashbox{ 40 }{ 42 }&\backslashbox{ 45 }{ 59 }&\backslashbox{ 57 }{ 64 }&\backslashbox{ 38 }{ 52 }&\backslashbox{ 39 }{ 55 }&\backslashbox{ 46 }{ 45 }&\backslashbox{ 57 }{ 57 }&\backslashbox{ 56 }{ 58 }&\backslashbox{ 37 }{ 50 }&\backslashbox{ 56 }{ 59 }&\backslashbox{ \color{red}{\textbf{222}} }{ 68 }&\backslashbox{ 36 }{ 63 }\\\hline 
$\textrm{BM}_3$ & \backslashbox{ 42 }{ 52 }&\backslashbox{ 36 }{ 51 }&\backslashbox{ 39 }{ 55 }&\backslashbox{ 31 }{ 53 }&\backslashbox{ 45 }{ 56 }&\backslashbox{ 50 }{ 56 }&\backslashbox{ 33 }{ 58 }&\backslashbox{ 21 }{ 45 }&\backslashbox{ 43 }{ 64 }&\backslashbox{ 36 }{ 62 }&\backslashbox{ 34 }{ 63 }&\backslashbox{ 45 }{ 57 }&\backslashbox{ 41 }{ 62 }&\backslashbox{ 35 }{ 59 }&\backslashbox{ 35 }{ 55 }&\backslashbox{ 23 }{ 45 }&\backslashbox{ \color{red}{\textbf{321}} }{ \color{blue}{\textbf{78}} }&\backslashbox{ 90 }{ 29 }\\\hline 
$\textrm{GW}_2$ & \backslashbox{ 42 }{ 45 }&\backslashbox{ 65 }{ 59 }&\backslashbox{ 68 }{ 59 }&\backslashbox{ 34 }{ 55 }&\backslashbox{ 61 }{ 56 }&\backslashbox{ 51 }{ 54 }&\backslashbox{ 58 }{ 60 }&\backslashbox{ 65 }{ 66 }&\backslashbox{ 51 }{ 59 }&\backslashbox{ 59 }{ 54 }&\backslashbox{ 46 }{ 62 }&\backslashbox{ 49 }{ 63 }&\backslashbox{ 33 }{ 47 }&\backslashbox{ 47 }{ 55 }&\backslashbox{ 44 }{ 52 }&\backslashbox{ 59 }{ 57 }&\backslashbox{ \color{red}{\textbf{142}} }{ 65 }&\backslashbox{ 26 }{ 32 }\\\hline 
$\textrm{GW}_3$ & \backslashbox{ 46 }{ 58 }&\backslashbox{ 63 }{ 70 }&\backslashbox{ 58 }{ 58 }&\backslashbox{ 36 }{ 55 }&\backslashbox{ 54 }{ 67 }&\backslashbox{ 53 }{ 52 }&\backslashbox{ 51 }{ 63 }&\backslashbox{ 60 }{ 63 }&\backslashbox{ 52 }{ 60 }&\backslashbox{ 57 }{ 55 }&\backslashbox{ 50 }{ 53 }&\backslashbox{ 47 }{ 61 }&\backslashbox{ 44 }{ 51 }&\backslashbox{ 50 }{ 52 }&\backslashbox{ 38 }{ 45 }&\backslashbox{ 55 }{ 62 }&\backslashbox{ \color{red}{\textbf{185}} }{ 71 }&\backslashbox{ 1 }{ 4 }\\\hline 
\end{tabular}}
\end{center}

\begin{center}
\captionof*{table}{\fontsize{8pt}{8pt}\selectfont MIS-NWS}
\resizebox{\textwidth}{!}{%
\begin{tabular}{|l||*{18}{c|}}\hline 
 \backslashbox{$\alpha$}{$\mathcal{P}$}&0&1&2&3&4&5&6&7&8&9&10&11&12&13&14&15&16&None
\\\hline\hline
$\textrm{BM}_2$ & \backslashbox{ 60 }{ 69 }&\backslashbox{ 48 }{ 54 }&\backslashbox{ 59 }{ 53 }&\backslashbox{ 54 }{ 63 }&\backslashbox{ 43 }{ 52 }&\backslashbox{ 61 }{ 53 }&\backslashbox{ 53 }{ 45 }&\backslashbox{ 51 }{ 41 }&\backslashbox{ 48 }{ 54 }&\backslashbox{ 65 }{ 52 }&\backslashbox{ 51 }{ 68 }&\backslashbox{ 54 }{ 52 }&\backslashbox{ 56 }{ 45 }&\backslashbox{ 43 }{ 41 }&\backslashbox{ 59 }{ 61 }&\backslashbox{ 45 }{ 71 }&\backslashbox{ \color{red}{\textbf{116}} }{ 69 }&\backslashbox{ 34 }{ 57 }\\\hline 
$\textrm{BM}_3$ & \backslashbox{ 47 }{ 45 }&\backslashbox{ 46 }{ 65 }&\backslashbox{ 48 }{ 64 }&\backslashbox{ 51 }{ 42 }&\backslashbox{ 51 }{ 73 }&\backslashbox{ 53 }{ 76 }&\backslashbox{ 50 }{ 56 }&\backslashbox{ 41 }{ 64 }&\backslashbox{ 58 }{ 60 }&\backslashbox{ 42 }{ 52 }&\backslashbox{ 60 }{ 59 }&\backslashbox{ 50 }{ 49 }&\backslashbox{ 67 }{ 55 }&\backslashbox{ 49 }{ 52 }&\backslashbox{ 37 }{ 54 }&\backslashbox{ 31 }{ 50 }&\backslashbox{ \color{red}{\textbf{121}} }{ 55 }&\backslashbox{ 98 }{ 29 }\\\hline 
$\textrm{GW}_2$ & \backslashbox{ 68 }{ 62 }&\backslashbox{ 64 }{ 56 }&\backslashbox{ 58 }{ 57 }&\backslashbox{ 66 }{ 63 }&\backslashbox{ 51 }{ 40 }&\backslashbox{ 67 }{ 64 }&\backslashbox{ 61 }{ 72 }&\backslashbox{ 58 }{ 76 }&\backslashbox{ 57 }{ 68 }&\backslashbox{ 55 }{ 55 }&\backslashbox{ 54 }{ 51 }&\backslashbox{ 51 }{ 53 }&\backslashbox{ 45 }{ 52 }&\backslashbox{ 52 }{ 58 }&\backslashbox{ 57 }{ 53 }&\backslashbox{ 61 }{ 66 }&\backslashbox{ 47 }{ 32 }&\backslashbox{ 28 }{ 22 }\\\hline 
$\textrm{GW}_3$ & \backslashbox{ 64 }{ 60 }&\backslashbox{ 62 }{ 59 }&\backslashbox{ 60 }{ 65 }&\backslashbox{ 65 }{ 67 }&\backslashbox{ 53 }{ 53 }&\backslashbox{ 51 }{ 59 }&\backslashbox{ 59 }{ 69 }&\backslashbox{ 62 }{ 76 }&\backslashbox{ 56 }{ 55 }&\backslashbox{ 66 }{ 56 }&\backslashbox{ 59 }{ 60 }&\backslashbox{ 60 }{ 59 }&\backslashbox{ 46 }{ 54 }&\backslashbox{ 51 }{ 59 }&\backslashbox{ 64 }{ 57 }&\backslashbox{ 73 }{ 61 }&\backslashbox{ 45 }{ 26 }&\backslashbox{ 4 }{ 5 }\\\hline 
\end{tabular}}
\end{center}

\caption{Each cell contains how often each vertex-at-top rotation choice maximizes $\alpha$ and $\mathcal{P}$ for each warmstart. Problem Dataset is the $1000$ problem instances for each problem type (see Section~\ref{setup}). Outliers for $\alpha$ ($\mathcal{P}$) are values which are $2$ standard deviations above the mean, and are indicated with {\color{red}\textbf{red}} ({\color{blue}\textbf{blue}}).}
\label{HistData}
\end{figure*}

\begin{figure*}

\begin{center}
\captionof*{table}{\fontsize{8pt}{8pt}\selectfont Continous}
\vspace{-0.3cm}

\resizebox{0.6\textwidth}{!}{%
\begin{tabular}{|c||c|c|c|}\hline 
 \backslashbox{$\alpha$}{$\mathcal{P}$}&First&Last&None
\\\hline\hline
$\textrm{BM}_2$ & \backslashbox{$0.6958\pm0.0979$}{$0.0033\pm 0.0140$}&\backslashbox{$0.7817\pm0.0132$}{$0.0045\pm 0.0132$}&\backslashbox{$0.6619\pm0.1090$}{$0.0018\pm 0.0096$}\\\hline
$\textrm{BM}_3$ & \backslashbox{$0.6677\pm0.1000$}{$0.0016\pm 0.0075$}&\backslashbox{$0.7616\pm0.0078$}{$0.0028\pm 0.0078$}&\backslashbox{$0.7497\pm0.0663$}{$0.0031\pm 0.0109$}\\\hline
$\textrm{GW}_2$ & \backslashbox{$0.8555\pm0.1102$}{$0.1465\pm 0.2030$}&\backslashbox{$0.9315\pm0.2107$}{$0.1938\pm 0.2107$}&\backslashbox{$0.7622\pm0.1351$}{$0.0661\pm 0.1315$}\\\hline
$\textrm{GW}_3$ & \backslashbox{$0.8476\pm0.1121$}{$0.1361\pm 0.1954$}&\backslashbox{$0.9273\pm0.2037$}{$0.1809\pm 0.2037$}&\backslashbox{$0.6823\pm0.1299$}{$0.0268\pm 0.0671$}\\\hline\end{tabular}}
\end{center}
\begin{center}
\captionof*{table}{\fontsize{8pt}{8pt}\selectfont Discrete}
\resizebox{0.6\textwidth}{!}{%
\begin{tabular}{|c||c|c|c|}\hline 
 \backslashbox{$\alpha$}{$\mathcal{P}$}&First&Last&None
\\\hline\hline
$\textrm{BM}_2$ & \backslashbox{$0.6965\pm0.0990$}{$0.0035\pm 0.0151$}&\backslashbox{$0.7751\pm0.0175$}{$0.0058\pm 0.0175$}&\backslashbox{$0.6693\pm0.1077$}{$0.0029\pm 0.0127$}\\\hline
$\textrm{BM}_3$ & \backslashbox{$0.6678\pm0.1021$}{$0.0021\pm 0.0081$}&\backslashbox{$0.7556\pm0.0125$}{$0.0037\pm 0.0125$}&\backslashbox{$0.7449\pm0.0631$}{$0.0041\pm 0.0153$}\\\hline
$\textrm{GW}_2$ & \backslashbox{$0.8501\pm0.1176$}{$0.1668\pm 0.2169$}&\backslashbox{$0.9311\pm0.2182$}{$0.2238\pm 0.2182$}&\backslashbox{$0.7585\pm0.1359$}{$0.0676\pm 0.1203$}\\\hline
$\textrm{GW}_3$ & \backslashbox{$0.8437\pm0.1178$}{$0.1562\pm 0.2114$}&\backslashbox{$0.9276\pm0.2139$}{$0.2120\pm 0.2139$}&\backslashbox{$0.6687\pm0.1281$}{$0.0291\pm 0.0738$}\\\hline\end{tabular}}
\end{center}
\begin{center}
\captionof*{table}{\fontsize{8pt}{8pt}\selectfont TSP}
\resizebox{0.6\textwidth}{!}{%
\begin{tabular}{|c||c|c|c|}\hline 
 \backslashbox{$\alpha$}{$\mathcal{P}$}&First&Last&None
\\\hline\hline
$\textrm{BM}_2$ & \backslashbox{$0.8881\pm0.0534$}{$0.0008\pm 0.0066$}&\backslashbox{$0.9509\pm0.0090$}{$0.0012\pm 0.0090$}&\backslashbox{$0.8784\pm0.0540$}{$0.0011\pm 0.0118$}\\\hline
$\textrm{BM}_3$ & \backslashbox{$0.8748\pm0.0544$}{$0.0007\pm 0.0069$}&\backslashbox{$0.9471\pm0.0079$}{$0.0010\pm 0.0079$}&\backslashbox{$0.9262\pm0.0329$}{$0.0011\pm 0.0081$}\\\hline
$\textrm{GW}_2$ & \backslashbox{$0.8783\pm0.0225$}{$0.0037\pm 0.0045$}&\backslashbox{$0.9488\pm0.0020$}{$0.0010\pm 0.0020$}&\backslashbox{$0.8929\pm0.0409$}{$0.0026\pm 0.0037$}\\\hline
$\textrm{GW}_3$ & \backslashbox{$0.8769\pm0.0196$}{$0.0035\pm 0.0038$}&\backslashbox{$0.9488\pm0.0007$}{$0.0007\pm 0.0007$}&\backslashbox{$0.8540\pm0.0455$}{$0.0013\pm 0.0023$}\\\hline\end{tabular}}
\end{center}
\begin{center}
\captionof*{table}{\fontsize{8pt}{8pt}\selectfont Portfolio Optimization}
\resizebox{0.6\textwidth}{!}{%
\begin{tabular}{|c||c|c|c|}\hline 
 \backslashbox{$\alpha$}{$\mathcal{P}$}&First&Last&None
\\\hline\hline
$\textrm{BM}_2$ & \backslashbox{$0.9661\pm0.0063$}{$0.0009\pm 0.0057$}&\backslashbox{$0.9684\pm0.0075$}{$0.0016\pm 0.0075$}&\backslashbox{$0.9633\pm0.0064$}{$0.0006\pm 0.0034$}\\\hline
$\textrm{BM}_3$ & \backslashbox{$0.9587\pm0.0082$}{$0.0003\pm 0.0026$}&\backslashbox{$0.9614\pm0.0023$}{$0.0005\pm 0.0023$}&\backslashbox{$0.9656\pm0.0060$}{$0.0005\pm 0.0041$}\\\hline
$\textrm{GW}_2$ & \backslashbox{$0.9987\pm0.0052$}{$0.9099\pm 0.2421$}&\backslashbox{$0.9993\pm0.2266$}{$0.9215\pm 0.2266$}&\backslashbox{$0.9645\pm0.0241$}{$0.2405\pm 0.3278$}\\\hline
$\textrm{GW}_3$ & \backslashbox{$0.9977\pm0.0067$}{$0.8551\pm 0.2858$}&\backslashbox{$0.9987\pm0.2665$}{$0.8748\pm 0.2665$}&\backslashbox{$0.9550\pm0.0201$}{$0.1060\pm 0.2090$}\\\hline\end{tabular}}
\end{center}
\begin{center}
\captionof*{table}{\fontsize{8pt}{8pt}\selectfont MIS-GNP}
\resizebox{0.6\textwidth}{!}{%
\begin{tabular}{|c||c|c|c|}\hline 
 \backslashbox{$\alpha$}{$\mathcal{P}$}&First&Last&None
\\\hline\hline
$\textrm{BM}_2$ & \backslashbox{$0.8196\pm0.0600$}{$0.0091\pm 0.0260$}&\backslashbox{$0.8698\pm0.0258$}{$0.0125\pm 0.0258$}&\backslashbox{$0.8063\pm0.0625$}{$0.0066\pm 0.0178$}\\\hline
$\textrm{BM}_3$ & \backslashbox{$0.7993\pm0.0625$}{$0.0047\pm 0.0124$}&\backslashbox{$0.8582\pm0.0157$}{$0.0078\pm 0.0157$}&\backslashbox{$0.8455\pm0.0465$}{$0.0087\pm 0.0184$}\\\hline
$\textrm{GW}_2$ & \backslashbox{$0.8850\pm0.0582$}{$0.0996\pm 0.1462$}&\backslashbox{$0.9222\pm0.1458$}{$0.1134\pm 0.1458$}&\backslashbox{$0.8590\pm0.0628$}{$0.0577\pm 0.0907$}\\\hline
$\textrm{GW}_3$ & \backslashbox{$0.8734\pm0.0612$}{$0.0874\pm 0.1424$}&\backslashbox{$0.9151\pm0.1406$}{$0.0980\pm 0.1406$}&\backslashbox{$0.8021\pm0.0663$}{$0.0236\pm 0.0455$}\\\hline\end{tabular}}
\end{center}
\begin{center}
\captionof*{table}{\fontsize{8pt}{8pt}\selectfont MIS-NWS}
\resizebox{0.6\textwidth}{!}{%
\begin{tabular}{|c||c|c|c|}\hline 
 \backslashbox{$\alpha$}{$\mathcal{P}$}&First&Last&None
\\\hline\hline
$\textrm{BM}_2$ & \backslashbox{$0.8078\pm0.0472$}{$0.0076\pm 0.0190$}&\backslashbox{$0.8218\pm0.0213$}{$0.0080\pm 0.0213$}&\backslashbox{$0.7902\pm0.0474$}{$0.0050\pm 0.0138$}\\\hline
$\textrm{BM}_3$ & \backslashbox{$0.7802\pm0.0525$}{$0.0049\pm 0.0145$}&\backslashbox{$0.8043\pm0.0163$}{$0.0058\pm 0.0163$}&\backslashbox{$0.8163\pm0.0426$}{$0.0085\pm 0.0203$}\\\hline
$\textrm{GW}_2$ & \backslashbox{$0.8851\pm0.0475$}{$0.1139\pm 0.1429$}&\backslashbox{$0.8672\pm0.1057$}{$0.0699\pm 0.1057$}&\backslashbox{$0.8493\pm0.0568$}{$0.0598\pm 0.0904$}\\\hline
$\textrm{GW}_3$ & \backslashbox{$0.8728\pm0.0488$}{$0.0954\pm 0.1326$}&\backslashbox{$0.8528\pm0.0873$}{$0.0511\pm 0.0873$}&\backslashbox{$0.7899\pm0.0630$}{$0.0225\pm 0.0465$}\\\hline\end{tabular}}
\end{center}
\caption{Each cell contains the mean and standard deviation of $\alpha$ and $\mathcal{P}$ for each warmstart/vertex-at-top rotation choice. Problem Dataset is the $1000$ problem instances for each problem type (see Section~\ref{setup}).}
\label{CompData}
\end{figure*}

\begin{figure*}

    \vspace{-.5cm}
    \foreach \i in {1,...,6} {
        \includegraphics[width=.85\linewidth,page=\i]{Figures/Full_Comparison_All.pdf}
    }
    \vspace{-.3cm}
    \caption{$(\alpha,\mathcal{P})$ data for $\textrm{GW}_{2}$, $\textrm{GW}_{3}$, $\textrm{BM}_{2}$, and $\textrm{BM}_{3}$ over $p$ for the $10$ continuous random QUBO, discrete random QUBO, TSP, Portfolio Optimization, MIS-GNP, and MIS-NWS problem instances. Datapoints are average values and shaded regions are $\pm 0.25$ standard deviations.}
    \label{Full_Comparison_All}
\end{figure*}

\end{document}